\documentclass{article}
\usepackage{fullpage}
\usepackage{amsmath,amssymb,amsthm}
\usepackage{graphicx}
\usepackage{natbib}
\usepackage[hidelinks]{hyperref}

\newcommand{\R}{\mathbb{R}}

\newcommand{\E}{\mathbb{E}}
\DeclareMathOperator{\var}{var}
\DeclareMathOperator{\cov}{cov}

\newcommand{\given}{\,\vert\,}
\newcommand{\st}{\,\colon\,} % such that

\newcommand{\one}{\mbox{\textbf{1}}}

\newcommand{\overbar}[1]{\mkern 1.5mu\overline{\mkern-1.5mu#1\mkern-1.5mu}\mkern 1.5mu}

\newtheorem{definition}{Definition}
\newtheorem{lemma}{Lemma}

\newtheorem{theorem}{Theorem}
\newtheorem{example}{Example}

\newcommand{\anc}{\mathcal{A}}
\newcommand{\ancs}{\mathcal{A}}
\newcommand{\pedigree}{\mathcal{P}}
\newcommand{\sites}{\mathcal{G}}
\newcommand{\muts}{M}
\renewcommand{\path}[2]{{{[#1 \leftrightarrow #2]}}}

\newif\ifdetails
\detailsfalse
\newcommand{\details}[1]{\ifdetails{#1}\else{}\fi}

\begin{document}

\title{An empirical approach to demographic inference with genomic data}
\author{Peter L.\ Ralph}

\maketitle

\begin{abstract}

  % The Ancestral Recombination Graph (ARG) is the object 
  % that adds information about genetic relationships 
  % (i.e., the outcomes of recombination in each meiosis)
  % to the genealogical relationships of the population pedigree.
  Inference with population genetic data
  usually treats the population pedigree as a nuisance parameter,
  the unobserved product of a past history of random mating.
  However, the history of genetic relationships in a given population
  is a fixed, unobserved object,
  and so an alternative approach is to treat this network of relationships
  as a complex object we wish to learn about,
  by observing how genomes have been noisily passed down through it.
  This paper explores this point of view,
  showing how to translate questions about population genetic data into
  calculations with a Poisson process of mutations on all ancestral genomes.
  This method is applied to give a robust interpretation to the $f_4$ statistic used to identify admixture,
  and to design a new statistic 
  that measures covariances in mean times to most recent common ancestor between two pairs of sequences.
  The method more generally interprets population genetic statistics in terms of sums of specific functions
  over ancestral genomes,
  thereby providing concrete, broadly interpretable interpretations for these statistics.
  This provides a method for describing demographic history without simplified demographic models.
  More generally,
  it brings into focus the population pedigree, 
  which is averaged over in model-based demographic inference.

\end{abstract}

\section*{Introduction}

For most of the history of the field,
geneticists made great progress with data that comprise
a tiny fraction of all that is in the genome.
Even the hundreds of thousands of markers on modern genotyping chips make certain inferences impossible
due to ascertainment of these markers.
Today, whole-genome sequence is being used to make striking new discoveries
about health and human history \citep[e.g.][]{bycroft2017genomewide,haak2015massive};
such data will soon be commonplace as well in nonmodel organisms.
The scale and unbiased nature of these data opens up new opportunities,
but will also require new methods,
as many assumptions made previously may no longer be necessary.

%% intro for stats journal
To do this, it is helpful to begin with the process that generated these data.
The rules by which the majority of genetic material 
is inherited by diploid, sexually reproducing organisms are simple:
each new individual carries two distinct copies of the genome 
-- one from mom, one from dad --
and each of these copies is formed by recombining different parts 
of the parent's two copies so as to make a whole.
The inherited copies may be modified by mutation,
which provides the raw material not only for evolution,
but also the means for us to learn about those genomes' histories.

The network describing parent-offspring relationships between past and present members of a population is known as the \emph{population pedigree};
adding an encoding of the genetic relationships 
-- which parts of each genome was inherited from which copy of the parental genomes --
produces the \emph{ancestral recombination graph} \citep[ARG, described by][]{griffiths1997ancestral}.
Supposing that any particular position in one of those copies of the genome 
was inherited from a single one of the two ancestral copies
-- for the most part true --
then at each point on the genome these genetic relationships form a tree 
(known as the \emph{gene tree} or \emph{gene genealogy}) 
that is embedded within the larger pedigree.
This information can be efficiently represented as a \emph{tree sequence} \citep{kelleher2016efficient},
and can now be produced on a genomic scale from population simulations \citep{kelleher2018efficient,haller2018treesequence}. 
Just as the population pedigree can be thought of as constraining the collection of gene trees,
so the ancestral recombination graph constrains the patterns of concordance and discordance of genetic variants 
created by mutation within each population.

Inference with population genetic data usually works
by specifying a \emph{population model} 
-- e.g., a randomly mating population of fixed size $N$ --
that determines a probability distribution on the population pedigree,
and then assuming that mutation and recombination occurs independently within this
\citep{ewens2004mpg}.
This paper deals with the intermediate layer --
if we take the population pedigree, or even the ancestral recombination graph, as \emph{fixed},
then what can we learn about it?
What \emph{are} the statistics of population genetics telling us about it?
This is partly a matter of philosophy:
should the population pedigree that has actually occurred be treated as 
an instantiation of a random process, with the goal of inferring parameters of that process?
Or, is it better to treat the population pedigree as a nonrandom but complex and partially observed object
that we wish to estimate summary statistics of?
The former view may be most appropriate if we have one or several well-defined demographic models we believe are close to the truth,
but the latter is useful for making robust inferences in the absence of any realistic demographic model,
as well as for identifying whether we have power to distinguish such models.

The approach of this paper is simple:
Assuming that mutations occur independently of demography and recombination (and are hence neutral),
these can be treated as a Poisson process on the set of all ancestral genomes, 
and most population genetics statistics can be written as 
integrals of specific functions against this Poisson process.
By treating the ancestral recombination graph as fixed,
the means and variances of these statistics can be found in terms of integrals of the same functions
against the mutation rate,
and can therefore be interpreted as descriptive statistics of the 
unknown, but real, ancestral recombination graph.
Stochastic models of population genetics generally include 
demography,
recombination,
and mutation:
this paper takes the outcomes of the first two as fixed,
dealing only with the randomness of mutation.
More information is obtainable by also treating recombination as random,
but it is useful to begin with only mutation,
especially as 
recombination itself is not directly observed, 
and many commonly used population genetics statistics do not model recombination 
(e.g., the site frequency spectrum).
% The addition of recombination will be described in a companion paper.

The remainder of this paper is structured as follows.
After a review of some of the literature,
two main results are presented in non-technical terms,
given for their inherent interest as well as 
examples of the general method of calculation that is
presented in the subsequent section.
Next come proofs of the main results,
followed by some more discussion of the results and future work.

%%%%%%% %%%%%%%
\subsection*{Background}

The main point of this paper is 
to see what can be learned about the population pedigree,
or more properly the ARG,
by taking an empirical point of view,
i.e., treating it as a fixed, but unknown object.
Others have taken this approach, albeit usually less explicitly.
\citet{slatkin1991inbreeding} 
interprets probabilities of identity by descent and $F_{ST}$ as summaries of the underlying gene trees;
the discussion is framed in terms of randomly mating populations,
but much of the discussion is applicable without thinking of the trees as random. 
\citet{mcvean2002genealogical,mcvean2009genealogical}
continued in this vein,
explaining what two common population genetics statistics,
linkage disequilibrium and the principal components decomposition,
have to say about these trees.
\citet{wakeley2012genealogies} begins from a similar philosophical point, 
that the pedigree should be treated as a fixed parameter,
and goes on to show that data generated under a fixed pedigree
differ very little from those generated under a standard coalescent model
(which does not incorporate correlations due to a fixed pedigree).

Formal methods describing expectations of gene transmission in a known pedigree
have long been used by plant and animal breeders.
In particular, 
\citet{wright1922coefficients}'s method of ``path coefficients''
treated the pedigree as a known object, averaging over recombination and segregation.
These methods are related to a parallel issue arising in genetic association studies
\citep[reviewed by][]{speed2014relatedness}
where relatedness of samples is a major confounding factor,
and alternative approaches can be seen as estimating the actual degree of genetic relatedness
or as averaging over it.

The empirical point of view is much more natural in the context of phylogenetics,
where between most groups of species, most loci have nearly the same genealogical relationship,
and the goal is to directly infer this common tree.
For instance,
\citet{gillespie1979evolutionary} pointed out that variation in coalescence time contributed to variance in between-species divergence,
and \citet{edwards2000perspective} further partition the variance in divergence 
into components induced by the substitution process and by the demographic process. 
% \citet{pluzhnikov1996optimal} also described,
% in a model of random mating, the variance of Watterson's and Tajima's estimates of $\theta$
% expected along the genome as a function of recombination rate.
Similarly, studies of nonrecombining loci often focus on the single underlying tree.
The tree that describes relationships between, say, mitochondria
provides a valuable look into population history and dynamics \citep{avise1987intraspecific,templeton1995separating},
and is sometimes of independent interest \citep{cann1987mitochondrial,vigilant1991african}.
\citet{tang2002frequentist} studied the ``frequentist'' approach 
to estimating time to most recent common ancestor (TMRCA) for modern human Y-chromosomes.
% includes "frequentist" mean and variance for T, including summed over all pairs who diverged at the root; 
\citet{hudson2007variance} continued this work, finding the variance of 
\citet{thomson2000recent}'s TMRCA estimator (which is based on the density of segregating derived mutations).
Here, we trade precise knowledge of one tree for vague knowledge about many, 
and are aided by the theoretical tools of phylogenetics \citep{semple2003phylogenetics}.

Early examples of using whole-genome data to estimate an empirical quantity of the population pedigree
were the inferences of gene flow between diverged species
by \citet{kulathinal2009genomics} (between \textit{Drosophila} species),
and by \citet{green2010draft} (Neanderthal to modern humans).
A wider family of similar statistics were subsequently developed in the human population genetics literature 
\citep[e.g.,][]{durand2011testing,patterson2012ancient},
and have formed the basis for many new insights into human history and evolution.
In recent years, a substantial body of work has used these statistics,
often applied to ancient DNA from ancient hominid remains,
to add new dimensions to our understanding of how modern humans evolved and colonized the globe
\citep[e.g.,][]{damgaard2018ancient,mccoll2018prehistoric}.
One of the main tools in this work is the $f_4$ statistic,
which is studied in detail below.
Many properties of these statistics are described in \citet{peter2016admixture},
but the general behavior of these statistics under complex models of population history remains unclear.
Understanding these in a more general context was a major impetus for the present work.

%%%%% %%%%%
\section*{Main results}

The purpose of this paper is to outline a general method of calculation,
the utility of which is best communicated with examples.
First, a few assumptions.
I assume that at each site there are no more than two variants seen
(the proportion of triallelic sites within populations is extremely small, but see \citet{jenkins2011effect});
and that at each site a reference allele has been chosen;
``allele frequencies'' are frequencies of the other, alternate allele.
However, all statistics developed here are invariant under relabeling of alleles.

I also make the standing assumption of low per base pair mutation rate  --
precisely, the infinite sites model of mutation \citep{watterson1975number},
under which each new mutation occurs at a distinct site,
and is hence observable.
This assumption is quite common in population genetics, 
where the typical time separating samples is small enough to make this assumption a good one
(but I return to this below).

As a first example, consider nucleotide diversity,
which for a set $X$ of sampled genomes is denoted $\pi(X)$
and is defined to be the mean proportion of genotyped sites
at which two genomes chosen randomly, without replacement from $X$ differ.
If each mutation occurs at a unique site, two genomes $x$ and $y$
differ at a site if and only if there has been a mutation at that site
in some ancestor from which either $x$ or $y$, but not both, has inherited that site.
In other words, sites that contribute to $\pi(X)$ through the comparison of $x$ and $y$ are those
at which a mutation has occurred on the branches leading from $x$ and $y$ back to their common ancestor.
If mutations occur with probability $\mu$ per generation and per site, 
nucleotide diversity is therefore an estimator of the mean time to most recent
common ancestor -- more precisely, $\pi(X) \approx 2 \mu \bar t_X$,
where $\bar t_X$ is the number of generations from two random samples from $X$
back to their common ancestor,
averaged over choices of samples and sites in the genome.
How precisely does $\pi(X)$ estimate this mean time to most recent common ancestor?
Suppose for simplicity that there are only two genomes in the set $X$.
Randomness only enters through the addition of mutations:
if we compute $\pi(X)$ averaging across $G$ sites,
then $G \pi(X)$ is Poisson with mean $2\mu \bar t_X G$;
and so the standard deviation of $\pi(X)$ is $\sqrt{2 \mu \bar t_X / G}$. 
Genomic data should give very accurate estimates of $\bar t_X$
in regions where the mutation rate is constant.

Next, let's apply the same intuition to
the $f_4$ statistic \citep{patterson2012ancient,peter2016admixture}:
if $X$, $Y$, $U$, and $V$ are four sets of sampled genomes,
and $p_X(\ell)$ is the allele frequency among the samples in $X$ at site $\ell$, etcetera,
then the $f_4$ statistic of these samples is 
\begin{align} \label{eqn:f4_definition}
  f_4(X,Y;U,V) = \frac{1}{G} \sum_{\ell=1}^G ( p_X(\ell) - p_Y(\ell) )( p_U(\ell) - p_V(\ell) )
\end{align}
where the sum is over some set of $G$ sites.
Notice that $f_4(X,Y;U,V) = -f_4(Y,X;U,V)$.
This implies that if the relationship of $X$ and $Y$ to $U$ and $V$ is symmetric,
then the expected value of the statistic is zero.
In particular, it is zero under a tree-based model of populations,
in which $X$ and $Y$ diverged from a common ancestral population more recently than their common ancestors with $U$ or $V$,
and there was no subsequent gene flow from $U$ or $V$.
\citet{green2010draft} and \citet{reich2010genetic} used these statistics
to provide evidence for gene flow from archaic hominids into modern humans.
For example,
taking $X$ to be a sample of sub-Saharan Africans, $Y$ a sample of western Europeans, $U$ a Neanderthal, and $V$ a chimpanzee,
a significantly negative value
implies that European allele frequencies are more correlated with Neanderthal
than expected under a model of no admixture;
and given archaeological evidence, 
this supports interbreeding between Neanderthals and modern humans after leaving Africa.
However, it is not at first clear how the statistic should perform under more realistic population histories
(e.g., asymmetric gene flow between sub-Saharan Africa and Mediterranean populations).

The $f_4$ statistic is usually motivated in the context of a particular \emph{population model}
of randomly mating populations that may split or merge with each other
(so their history is described by a simple admixture graph).
Is it still useful even if history is not well-described
as occasional interactions between a few randomly mating populations?
A thought experiment is useful at this point.
Suppose that in addition to nucleotide sequences, 
we also had access to any sort of genealogical information we might want.
Would the genomes still be useful in fitting the model?
The answer is \emph{no} --
since the model is neutral, genomes only provide information about past history
insofar as they shed light on past genealogical relationships.
In other words, the genealogies, if we had them, 
would be a sufficient statistic for inference of these population models.
So, the $f_4$ statistic 
describes something about those genealogies
that is informative about the presence or absence of gene flow between random mating populations.
For understanding the $f_4$ statistic 
-- both under the model above and in other contexts --
it is helpful to describe this ``something about those genealogies''.
Imagine annotating every possible ancestral individual
(i.e., everyone who ever lived) 
with how much each set of samples inherit from that individual.
The results below provide a (model-free) interpretation of the $f_4$ statistic
as an estimator of a summary statistic of these annotations 
(illustrated in Figure \ref{fig:f4_ex}).

\begin{theorem}[$f_4$ statistic]
    \label{thm:f_4}
  Let $X$, $Y$, $U$, and $V$ be collections of $n_X$, $n_Y$, $n_U$, and $n_V$ chromosomes, respectively,
  that have been genotyped at $G$ sites.
    Under the infinite-sites model of mutation, the following equivalences hold.
 \textbf{(I)} For each position $\ell$ in the genome and each ancestral individual's chromosome, $a$, 
    let $F_X(a,\ell)$ denote the proportion of the samples in $X$ that have inherited from $a$ at position $\ell$, 
  and similarly for $F_Y$, $F_U$, and $F_V$.
  Then 
  \begin{align} \label{eqn:f4_expectation}
    \E[ f_4(X,Y;U,V) ] = \frac{1}{G} \sum_{\ell=1}^G \sum_a \mu(\ell) ( F_X(a,\ell) - F_Y(a,\ell) )( F_U(a,\ell) - F_V(a,\ell) ) ,
  \end{align}
  where $\mu(\ell)$ is the per-generation mutation rate at site $\ell$,
  and the sum over $a$ is the sum over all ancestral chromosomes.
 \textbf{(II)} Equivalently, for any set of samples $(x,y,u,v)$, 
  let $S_{xu|yv}(\ell)$ denote
  the lengths of any branches in the tree at position $\ell$
  between the most recent common ancestor of $x$ and $u$ 
  and any ancestor of $y$ or $v$,
  or between the most recent common ancestor of $y$ and $v$ and any ancestor of $x$ or $u$.
  (This is the length of the internal branch in the unrooted tree at position $\ell$
  if it has topology $xu|yv$, and is zero otherwise.) 
  Then
  \begin{align} \label{eqn:f4_expectation_branches}
    \E[ f_4(X,Y;U,V) ] = \frac{1}{n_X n_Y n_U n_V} \sum_{x, y, u, v} \frac{1}{G} \sum_{\ell=1}^G \mu(\ell) ( S_{xu|yv}(\ell) - S_{xv|yu}(\ell) ) ,
  \end{align}
  where the sum is over all choices of $x \in X$, $y\in Y$, $u \in U$, and $v \in V$.
  Furthermore,
  \begin{align} \label{eqn:f4_variance}
    \var[ f_4(X,Y;U,V) ] 
    \le 
    \frac{1}{n_X n_Y n_U n_V} \sum_{x, y, u, v} \frac{1}{G^2} \sum_{\ell=1}^G \mu(\ell) ( S_{xu|yv}(\ell) + S_{xv|yu}(\ell) ) .
  \end{align}
\end{theorem}

\begin{figure}
    \begin{center}
        \includegraphics[width=\textwidth]{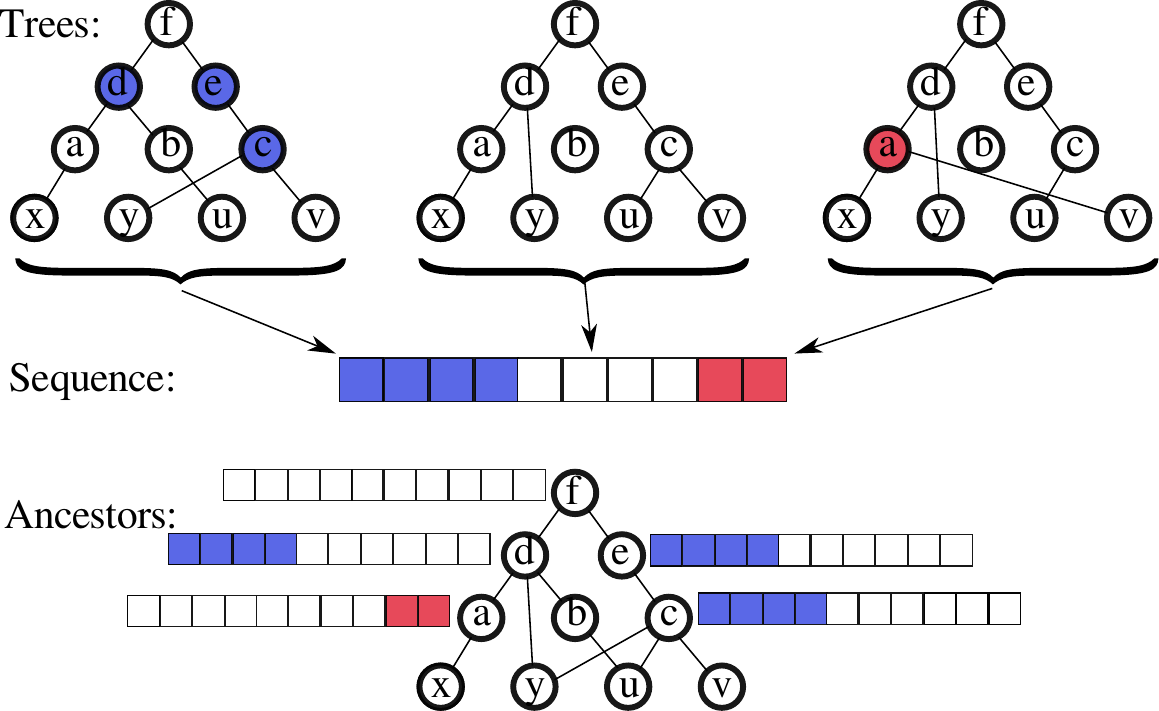}
    \end{center}
    \caption{
        Depiction of calculating expected $f_4(x,y;u,v)$ from the pedigree
        relating the four chromosomes $x$, $y$, $u$, and $v$.
        \textbf{Top} shows the three distinct trees found along the chromosome,
        and \textbf{middle} shows the locations on the 10-base-pair chromosome
        where those trees are found.
        Nodes in the trees are colored by their expected contribution to $f_4$ 
        -- any mutation occurring in blue nodes at that position will contribute ${}+\mu/10$ to 
        $f_4(x,y;u,v0$, while mutations occurring in red nodes would contribute ${}-\mu/10$,
        so that $\E[f_4(x,y;u,v)] = \mu (4 \times 3 + 0 - 2 \times 1) = 10\mu$.
        \textbf{Bottom} shows another way of computing this:
        individuals are provided with chromosomes labeled 
        according to how they contribute to $f_4(x,y;u,v)$.
        Instead of summing over trees, we can sum over ancestors, so that
        so that $\E[f_4(x,y;u,v)] = \mu (4 + 4 + 4 - 2) = 10\mu$.
        \label{fig:f4_ex}
    }
\end{figure}

Note that $F_X(a,\ell)$ is \emph{unknown} but quantifies the genetic inheritance of $X$,
and so equation \eqref{eqn:f4_expectation}
interprets $f_4(X,Y;U,V)$ as a descriptive statistic of demographic history.
Once formulated using the general-purpose theory,
the proof of this comes down to straightforward calculations with Poisson processes,
which are carried out below.

As a final example, the theory can be used to design a statistic,
to estimate the covariance in times to most recent common ancestor between two pairs of chromosomes.
This statistic has not to my knowledge been described in the literature.
Concretely, let $(x_1,x_2)$ and $(y_1,y_2)$ be two pairs of sampled chromosomes,
and for each position $\ell$ in the genome,
let $t_{x}(\ell)$ denote the total number of ancestors 
from whom either $x_1$ or $x_2$, but not both, have inherited at $\ell$,
and likewise for $t_y(\ell)$.
(So, $t_x(\ell)$ is the length of the tree, in number of meioses, between $x_1$ and $x_2$ at $\ell$.)
Suppose the genome is divided into nonoverlapping pieces $\{W_1, \ldots, W_n\}$,
and we wish to estimate covariances of mean values of $t_x$ and $t_y$, scaled by mutation rate,
across these chunks.
(The scaling by mutation rate could be removed if it was constant across sites, or known,
but I am leaving it in for realism.)
In other words, we want the covariance between 
$\overbar{\mu t}_{k,x} = \frac{1}{|W_k|} \sum_{\ell \in W_k} \mu_\ell t_x(\ell)$,
where $\mu_\ell$ is the mutation rate at site $\ell$,
and the corresponding quantity for $y$.
For convenience, also let $\overbar{\overbar{\mu t}}_x = (1/G) \sum_{\ell=1}^G \mu_\ell t_x(\ell)$ 
denote the mean time to common ancestor across the entire genome,
scaled by mutation rate;
and likewise for $\overbar{\mu t}_{k,y}$ and $\overbar{\overbar{\mu t}}_y$.

\begin{theorem}[Empirical covariance of mean coalescence times]
  \label{thm:mean_coal_cov}
  Let $N_k(x)$ denote the number of mutations that differentiate $x_1$ and $x_2$ in window $W_k$,
  $N_k(y)$ the same for $y_1$ and $y_2$,
    and let $N_k(x,y)$ denote the number of mutations that do both, 
  i.e., the number of mutations that differentiate between $x_1$ and $x_2$ as well as between $y_1$ and $y_2$.
  Define 
  $Z_{jk} = N_j(x) N_k(y) / |W_j| |W_k|$
  for $j \neq k$, and
    $Z_{jj} = (N_j(x) N_j(y) - N_j(x,y)) / |W_j|^2$. 
  The statistic
  \begin{align} \label{eqn:coal_cov_defn}
    C_{x,y} := \frac{1}{n} \sum_{k=1}^n \left( Z_{kk} - \frac{1}{n} \sum_{j=1}^n Z_{jk} \right) 
  \end{align}
  estimates the covariance in mean mutation-scaled coalescent times:
  % Continuing the notation from above, 
  % $\bar \theta_{k,x} = \frac{1}{|W_k|} \sum_{\ell \in W_k} \theta_x(\ell)$,
  % define $\bar \theta^k_y$ likewise,
  % and $\bar \theta_{k,xy} = \frac{1}{|W_k|} \sum_{\ell \in W_k} \theta_x(\ell)\theta_y(\ell)$.
  % Also define $\bar{\bar \theta}_x = \frac{1}{n} \sum_{k=1}^n \bar \theta_{k,x}$.
  % Similarly, define $\overbar{\mu t}_{k,x} = \frac{1}{|W_k|} \sum_{\ell \in W_k} \mu_\ell t_x(\ell)$.
  % \label{thm:coal_covariance}
  % The statistic
  % \begin{align} \label{eqn:coal_cov_defn}
  %   C_{x,y} := \frac{1}{n} \sum_{k=1}^n \left( \bar\theta_{k,x}\bar\theta_{k,y} - \frac{1}{|W_k|} \bar\theta_{k,xy} \right)
  %   - \left( \bar{\bar \theta}_x \bar{\bar \theta}_y - \frac{1}{n} \sum_{k=1}^n \frac{1}{|W_k|} \bar{\bar \theta}_{k,xy} \right) 
  % \end{align}
  under the infinite-sites model of mutation,
  \begin{align}
    \E\left[ C_{x,y} \right] = \frac{1}{n} \sum_{k=1}^n \left( \overbar{\mu t}_{k,x} - \overbar{\overbar{\mu t}}_x \right)\left( \overbar{\mu t}_{k,y} - \overbar{\overbar{\mu t}}_y \right),
  \end{align}
  with variance 
  \begin{align}
    \var\left[ C_{x,y} \right] \le \frac{ \epsilon^3 }{ n |W| },
  \end{align}
  if $\epsilon = \max_k \{ \max( \overbar{\mu t}_{k,x},\overbar{\mu t}_{k,y} )\}$ and $|W| = \max_k |W_k|$.
\end{theorem}

Now, I use the theory of Poisson processes
to generalize these examples.

%%%%% %%%%%
\section*{A formal calculus of mutations on the ARG}

\paragraph{Genealogy}
Ambitiously,
the object of study is the collection of genetic relationships between every individual ever in the population, past and present.
To simplify the discussion of relationships between diploids,
consider relationships on only a single chromosome,
and let $\anc$ denote the collection of all copies of this chromosome in any individual alive at any point 
back to some arbitrarily distant time.
The \emph{population pedigree}, denoted $\pedigree$, 
is the directed acyclic graph with vertex set $\anc$
that has an edge $a \to b$ if $a$ is a parent of $b$,
so that each vertex in $\anc$ has in-degree equal to two.
(Although this is a different object than the pedigree relating the diploid individuals,
rather than their constituent chromosomes,
the two are simple transformations of each other.)
Under the standing assumption of the infinite sites model, 
no two mutations fall in the same location,
and so locations on the chromosome are indexed by the continuous interval $\sites = [0,G]$,
where $G$ is the total length of the chromosome.
At each site in $\sites$ all ancestral chromosomes are related by a tree,
embedded in the population pedigree,
and this collection of trees is equivalent to the ARG.

\paragraph{Mutation}
Under the infinite sites model, the mutation process is Poisson:
concretely, 
the sites at which ancestral chromosome $a \in \anc$ differs from their parent (i.e., has a mutation) 
% at genomic locations $\{m_{a,1}, \ldots, m_{a,n_a}\} \subset \sites$,
are distributed as the points of an inhomogeneous Poisson process on $\sites$.
This Poisson process of mutations has intensity measure given by the function $\mu$, 
where $\mu(\ell)$ is the mutation rate at $\ell$.
To study the inheritance of mutations through the pedigree,
these are next collected into a Poisson process on all ancestral chromosomes:
formally, $\muts$ % $ = \sum_{a \in \anc} \sum_{k=1}^{n_a} \delta(a) \otimes \delta_{m_{a,k}}$ 
is the point measure on $\anc \times \sites$ formed by putting one point mass at each mutation.
Equivalently, $\muts$ counts mutations:
for any region of the chromosome $[i,j] \subset \sites$
and any set of ancestral chromosomes $A \subset \ancs$,
the value $\muts(A \times [i,j])$ is the total number of mutations that occurred in the region $[i,j]$ 
on any of the chromosomes in $A$.

Under the assumption that mutation occurs independently of the pedigree,
the expected number of mutations occurring in any collection of pieces of ancestral chromosomes
is simply the sum of $\mu(\ell)$ across these pieces.
Formally,
$\muts$'s mean measure has density $\mu$ with respect to $n \otimes d\ell$,
where $n$ is counting measure on $\anc$ and $d\ell$ is Lesbegue measure on the chromosome.
Below, I will abuse notation slightly and write $d\mu$ for this measure,
so that for any (measurable) subset of ancestral genomes $S \subset \ancs \times \sites$,
\begin{align}
  \E[ \muts(S) ] = \int_S d\mu(a,\ell)  .
\end{align}

The key simplifying assumption of the infinite sites model is that every mutation is \emph{observable}:
if two ancestral chromosomes $a$ and $b$ have both inherited a region of genome $[i,j]$ 
from their most recent common ancestor on this region,
then given the genomes of $a$ and $b$ on this segment, 
we also know the number of mutations that have occurred
in this segment in any of the ancestors going back to the common ancestor.
% In other words, if $[i,j]$ is coinherited by $a$ and $b$ from $\rho_{a,b}(i)$, 
% we can observe $\muts( T_{\{a,b\}}(i) \times [i,j] )$ by comparing the genomes of $a$ and $b$.

\paragraph{Additive statistics}
The formulation of mutations as a point measure on ancestral genomes
is designed to make a certain set of statistics easy to formulate and analyze,
namely, those that can be formed using
integrals of test functions against the mutation process.
Recall that since $\muts$ is a point measure on the set of all ancestral genomes, $\anc \times \sites$,
with unit mass at each ancestral mutation, 
then if $\chi$ is a function on $\anc \times \sites$, the integral of $\chi$ against $\muts$
is the sum of the values of $\chi$ at the locations of all mutations:
defining $\{ m_{a,k} \st 1 \le k \le n_a \}$ to be the locations of the mutations that occurred on ancestral chromosome $a$,
\begin{align}
  \int \chi(a,\ell) d\muts(a,\ell) = \sum_{a \in \anc} \sum_{k=1}^{n_a} \chi(a,m_{a,k}) .
\end{align}
It is easier to write this as simply $\int \chi d\muts$.
Here is a simple example (for which all this notation is unnecessary):

\begin{example}[Number of mutations]
  Given a sample of chromosomes $A \subset \anc$ let $\zeta_A(a,\ell) = 1$ if $a \in A$ and $\zeta_A(a,\ell) = 0$ otherwise.
  Then $\int \zeta_A d\muts$ counts the number of mutations appearing \emph{de novo} in the sampled chromosomes.
\end{example}

This is a simple prototype for what we would like to do more generally.
But, notice that $\int \zeta_A d\muts$ is unobservable given only the genomes in $A$ -- 
to identify new mutations, we need the parental genomes as well.
% \begin{definition}[Ancestrally additive statistics]
%   Let $\chi$ be a mapping from collections of $n$ samples to functions on ancestral genomes, i.e.,
%   for each collection $A_1, \ldots, A_n$ of subsets of $\anc$,
%   the function $\chi_{A_1, \ldots, A_n}(\cdot,\cdot)$ is real-valued and has bounded support on $\anc \times \sites$.
%   The \emph{ancestrally additive statistic} associated with $\chi$
%   is
%   \begin{align*} \label{eqn:additive_defn}
%     S_\chi(A_1, \ldots, A_n) = \int_{\anc \times \sites} \chi_{A_1, \ldots, A_n}(a,\ell) d\muts(a,\ell) .
%   \end{align*}
% \end{definition}
Statistics of this form work by 
first, based on the samples, constructing a function $\chi$ on ancestral genomes,
and then adding up the value that $\chi$ takes at the location of each mutation on any ancestral genome.
To be observable,
such statistics will need to satisfy certain conditions:
in particular, 
$\chi$ must be zero for any piece of ancestral genome not inherited by any sampled genome, or inherited by all of them.
This idea of observability under the infinite sites model
is encoded by the following set of special functions:

\begin{definition}[Paths]
  \textbf{(I)}
      For any two chromosomes $x$ and $y$ in $\anc$,
      the \emph{path from $x$ to $y$} is the function
      \begin{align}
        \path{x}{y}(a,\ell) = 
        \begin{cases}
          1 \qquad & \text{if $x$ or $y$, but not both, inherits from $a$ at $\ell$} \\
          0 \qquad & \text{otherwise} .
        \end{cases}
      \end{align}
  \textbf{(II)}
      For any two collections of chromosomes $X$ and $Y$,
      having $|X|$ and $|Y|$ elements each,
      the \emph{path from $X$ to $Y$} is the function
      \begin{align}
        \path{X}{Y}(a,\ell) = \frac{1}{C(X,Y)} \sum_{x,y} \path{x}{y}(a,\ell) ,
      \end{align}
      where the sum is over distinct pairs $x \in X$ and $y \in Y$ with $x\neq y$,
      and $C(X,Y)$ is the number of such pairs.
    These are depicted in figure \ref{fig:example_paths}.
\end{definition}

\begin{figure}
    \begin{center}
        \includegraphics[width=\textwidth]{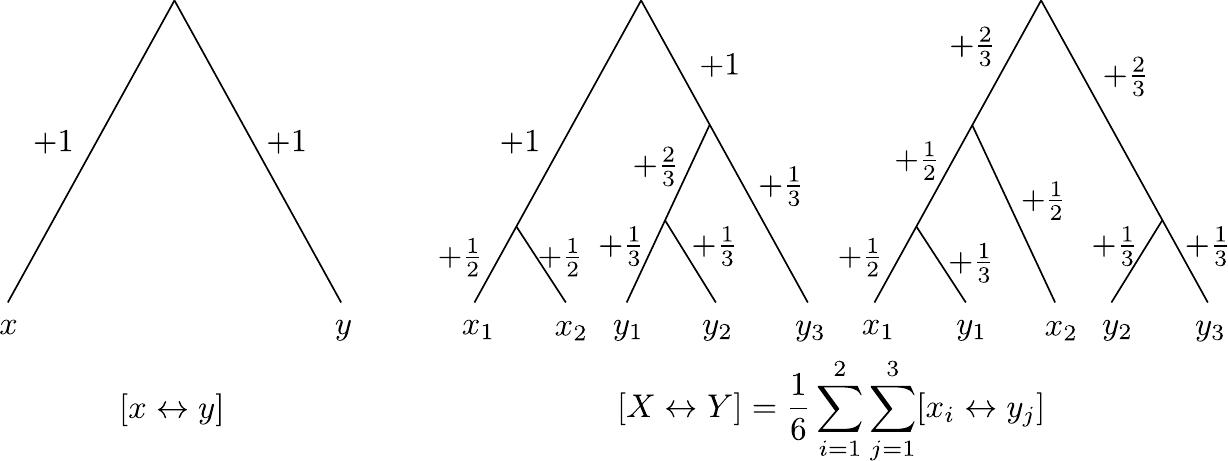}
    \end{center}
    \caption{
        A cartoon of values taken by the \emph{path} between \textbf{(left)} two sampled chromosomes $x$ and $y$,
        and \textbf{(right)} two groups of samples, $X=\{x_1,x_2\}$, and $Y=\{y_1,y_2,y_3\}$.
        The path between two single chromosomes is, marginally at each locus,
        the indicator of the path from each back to their common ancestor;
        but between larger groups the values it takes depends on the local topology
        of the tree relating the samples.
        \label{fig:example_paths}
    }
\end{figure}

In other words, the path $\path{x}{y}$ is the indicator function of the pieces of ancestral chromosomes
that are inherited by $x$ or $y$, but not both.
At each locus $\ell$, the function $\path{x}{y}(\cdot,\ell)$ on $\anc$ is the indicator function of 
the tree connecting $x$ and $y$ at $\ell$
(but notably excluding the root).
$\path{X}{Y}$ is the average of these indicators,
which gives the proportion of paths between pairs of samples from the two groups
that pass through each ancestor in the pedigree.

% The flow is similar, 
% but an alternative interpretation is that $\flow{X}{Y}(a,\ell)$ 
% is the proportion of members of $X$ that inherited from $a$ at $\ell$,
% minus the proportion of members of $Y$ that did.

\begin{example}[Sequence divergence]
  The \emph{sequence divergence} between two chromosomes $x$ and $y$, denoted $\pi(x,y)$, 
  is the mean density of distinguishing mutations;
  it is
  \begin{align}
    \pi(x,y) = \frac{1}{G} \int \path{x}{y} d\muts .
  \end{align}
  For a group of samples $X$ we define the average pairwise divergence similarly:
  \begin{align}
    \pi(X) 
      &= \frac{1}{|X|(|X|-1)} \sum_{x_1\neq x_2 \in X} \pi(x_1,x_2) \\
      &= \frac{1}{G} \int \left( \path{X}{X} \right) d\muts .
  \end{align}
\end{example}

\citet{tajima1983evolutionary} proposed $\pi(X)$ % Evolutionary relationship of DNA sequences in finite populations.
as an estimator for $\theta = 4 N_e \mu$,
under the model of a large, randomly mating population of diploids with constant mutation rate $\mu$
and effective population size $N_e$.
Indeed, $\E[\pi(X)] = \theta$, 
but the variance of $\pi(X)$ for a nonrecombining locus does not go to zero as the number of samples grows,
% ($\var[\pi(X)] = \frac{n+1}{3(n-1)} \theta + \frac{2(n^2+n+3)}{9n(n-1)} \theta^2$)
and so this is regarded as a poor estimator of $\theta$.
However, it is clear that 
the variance does go to zero if calculated with the gene tree fixed,
i.e., thought of as an estimator of the mean time to most recent common ancestor
(multiplied by twice the mutation rate).
Denoting by $\overbar{\mu t}$ the empirical mean mutation-rate-scaled time to most recent common ancestor,
the variance can be partitioned (as in \citet{edwards2000perspective}):
$\var[\pi(X)] = \E[ \var[\pi(X) \given \overbar{\mu t} ] ] + \var[ \E[ \pi(X) \given \overbar{\mu t} ] ]$.
The first term here goes to zero as the size of the sample increases,
but the second, due to randomness in the demographic process, does not.
The question of whether $\pi(X)$ is a good estimator, and for what,
becomes a question of whether we want to think of it as estimating a summary statistic of actual relationships
or a parameter in a certain stochastic model of demography.

\paragraph{Observable statistics}
For application to data, it only makes sense 
to consider statistics that can be calculated using a sampled set of genomes. 
As discussed above, under the infinite sites model
this means that the genome sequence of a pair of samples $x$ and $y$ on subset of the genome $L$
is determined by the number of mutations inherited by $x$ or $y$ but not both on $L$,
or $\sum_{a \in \anc} \sum_{\ell \in L} \path{x}{y}(a,\ell)$.
If $\one_L(a,\ell)=1$ for $\ell \in L$ and is zero otherwise,
this is $M(\one_L \path{x}{y})$,
so the set of observable statistics, given a collection of samples $A$,
is generated by the set of linear combinations of functions of this form,
i.e., products of indicators of a segment of genome with a path between two samples.
It is clear that we cannot learn about parts of the pedigree from which no samples have inherited,
or about anything occurring longer ago than the most recent common ancestor of the samples at each site,
but it is not clear how to formulate more generally what information about the ARG
is or is not obtainable from finite samples of chromosomes.

This notion of observability is made formal by the following definition:
\begin{definition}[Observable statistics]
  A function $\chi(a,\ell)$ on $\anc \times \sites$ is \emph{observable} given a collection of chromosomes $A$
  if it is in the algebra % is measurable with respect to the $\sigma$-algebra 
  generated by functions of the form
  $\one_{[u,v]}(a,\ell) \path{x}{y}(a,\ell)$, 
  where $[u,v]$ is an interval of the chromosome,
  and $x$ and $y$ are samples in $A$.
\end{definition}

There are further, unavoidable, limitations to what it is possible to learn from the data.
For instance, the sequence divergence $\pi(x,y)$ between $x$ and $y$
is an estimate of the mean of the empirical distribution of times to most recent common ancestors (TMRCA)
between $x$ and $y$,
multiplied by the mutation rate.
Theorem \ref{thm:mean_coal_cov} gives an estimator of the variance of the \emph{mean} TMRCA in windows.
However, it is not possible to estimate the variance of the distribution of \emph{single-site} TMRCA values
without somehow modeling dependencies between sites induced by recombination.
An easy way to see this is if each site has at most one mutation,
then without using inter-site dependencies,
the model is equivalent to one where the probability of a segregating mutation is \emph{constant},
equal to the mean.
To make further progress,
it is necessary (and reasonable) to make further assumptions,
that are beyond the scope of this paper.

%%%% %%%% %%%%
\subsection*{Moments}

Calculations with statistics in this formalism can be made with the help of
the following general formula for integrals of test functions
against a Poisson processes:

\begin{lemma}{[Generating function]}
For any test function $\phi: \anc \times \sites \to \R$
for which $\int \phi d\mu$ is absolutely convergent,
the statistic $\int \phi d\muts$ is well-defined, 
and 
\begin{align} \label{eqn:campbells_thm}
    \E\left[ \exp\left( \int \phi d\muts \right) \right] = 
        \exp\left( \int \left( e^{\phi} - 1 \right) d\mu \right) .
\end{align}
\end{lemma}

\begin{proof}
    This is Campbell's Theorem; see e.g., \citet{kingman-poisson-processes}.
    Recall that $\int \phi d\muts$ is the sum of the value that $\phi$ takes
    over all mutations occurring in $\muts$;
    ``well-defined'' means that this sum is absolutely convergent
    even if there are an infinite number of mutations.
\end{proof}

We immediately get, by differentiating this formula, the moments:

\begin{lemma}[Mean and covariance]
  \label{thm:mean_and_cov}
    Let $\phi$ and $\psi$ be test functions on $\anc \times \sites$
    for which $\int \phi d\mu$, $\int \psi d\mu$, and $\int \phi \psi d\mu$
    are all absolutely convergent.
    Then
    \begin{align}
        \E\left[ \int \phi d\muts \right] &= \int \phi d\mu \\
        \cov\left[ \int \phi d\muts,  \int \psi d\muts \right] &= \int \phi \psi d\mu .
    \end{align}
\end{lemma}

\begin{proof}
  These are a standard calculations, but we carry out the second for covariance for concreteness.
  Let $F(a,b) = \exp\left( \int (a \phi + b \psi) d\muts \right)$.
  % Then $\partial_a F(a,b) = \int \phi d\muts F(a,b)$,
  % so using \eqref{eqn:campbells_thm},
  % \begin{align}
  %   \E\left[ \int \phi d\muts \right] 
  %   &= \E \left[ \partial_a F(a,b) \vert_{a=b=0} \right] \\
  %   &= \int \phi e^{a\phi+b\psi} d\mu \exp\left( \int \left( e^{a\phi+b\psi} d\mu \right) \vert_{a=b=0} \\
  %   &= \int \phi d\mu .
  % \end{align}
  Using \eqref{eqn:campbells_thm},
  \begin{align}
    \partial_a \partial_b \log \E\left[ F(a,b) \right] \big \vert_{a=b=0}
    &= \int \phi \psi d\mu 
  \end{align}
  On the other hand, exchanging the expectation and the derivatives,
  \begin{align}
    \partial_a \partial_b \log \E\left[ F(a,b) \right] \big \vert_{a=b=0}
    &= \frac{ \E\left[ \partial_a \partial_b F(a,b) \right] - \E\left[ \partial_a F(a,b) \right] \E\left[ \partial_b F(a,b) \right] }{ \E[ F(a,b)^2 ] } \bigg \vert_{a=b=0} \\
    &= \E\left[ \int \phi \psi d\muts \right] - \E\left[ \int \phi d\muts \right] \E\left[ \int \psi d\muts \right] .
  \end{align}
\end{proof}

Expressions for more general moments,
or statistics whose expected values match desired quantities,
can be found using the following general formula:

\begin{lemma}{[General Moments]}
    \label{thm:moments}
  Let $\mathcal{S}$ be the set of upper triangular $n \times n$ matrices with entries in $\{0,1\}$ and whose columns sum to 1
  (i.e., ways of partitioning the set $\{1,2,\ldots,n\}$ into sets),
  and for $\sigma \in \mathcal{S}$ let $|\sigma|$ denote the number of nonzero rows of $\sigma$
  (i.e., the number of sets in the partition).
  Then for any collection of test functions $\phi_1, \ldots, \phi_n$ on $\anc \times \sites$
  for which the following integrals against $\mu$ are absolutely convergent,
  \begin{align}
    \E\left[ \prod_{i=1}^n \int \phi_i d\muts \right] &= \sum_{\sigma\in\mathcal{S}} \prod_i \int \left( \prod_j \phi_j^{\sigma_{ij}} \right) d\mu 
  \end{align}
  and
  \begin{align}
    \int \prod_{i=1}^n \phi_i d\mu &= \sum_{\sigma\in\mathcal{S}} (-1)^{1+|\sigma|} \E\left[ \prod_i \int \left( \prod_j \phi_j^{\sigma_{ij}} \right) d\muts \right]  .
  \end{align}
\end{lemma}

\emph{Note:} methods for efficient computation of mixed moments
using orthogonal polynomials
are presented in a general framework by \citet{peccati2011facts}.

\begin{proof}
    By induction, by differentiating either the generating function (for the first formula)
    or the log of the generating function (for the second) of $\sum_i \alpha_i \phi_i$.
\end{proof}

%%%%%

\section*{Application of the method}

We can now use the formalism developed above
to easily prove the theorems given in the Introduction.

%% F4
\begin{proof}[Proof of Theorem~\ref{thm:f_4}]

  Write $G_{x\ell}$ for the genotype of haplotype $x$ at position $\ell$,
  coded as `0' for the reference allele and `1' for the alternate allele.
  First note that
  \begin{align}
    f_4(X,Y;U,V) = \frac{1}{G} \sum_{\ell=1}^G \frac{1}{n_X n_Y n_U n_V} \sum_{x,y,u,v} (G_{x\ell} - G_{y\ell}) (G_{u\ell} - G_{v\ell}) ,
  \end{align}
  i.e., $f_4$ is the average product of differences between pairs of haplotypes chosen from $(X,Y)$ and from $(U,V)$,
  averaging over site in the genome and choices of pairs $(x,y)$ and $(u,v)$.
  For a given quadruple $(x,y,u,v)$,
  sites will add to this sum if there is a mutation shared by $x$ and $u$
  that $y$ and $v$ do not carry, or vice versa; 
  and sites will subtract from the sum if there is a mutation that is shared by $x$ and $v$ that $y$ and $u$ do not carry,
  or vice versa.
  The theorem for $f_4(x,y;u,v)$ then follows easily by writing the statistic
  as the difference of two Poisson random variables, divided by $G$.
  To carry this through using the notation here,
  note that the function $\path{x}{v} \path{y}{u} - \path{x}{u} \path{y}{v}$
  takes the value $+1$ on the internal branch of trees with unrooted topology $((x,u),(y,v))$,
  $-1$ on the internal branch of trees with unrooted topology $((x,v),(y,u))$,
  and is zero on the remaining topology $((x,y),(u,v))$;
  therefore, 
  \begin{align}
    f_4(X,Y;U,V) = \frac{1}{G} M \left( 
    % \frac{1}{n_X n_Y n_U n_V} \sum_{x,y,u,v} 
      % \left\{
        \path{X}{V} \path{Y}{U} - \path{X}{U} \path{Y}{V}
      % \right\}
    \right) .
  \end{align}
  The expected value of $f_4$ is therefore just the integral of 
   $\path{X}{V} \path{Y}{U} - \path{X}{U} \path{Y}{V}$
  against $\mu$;
  using linearity this can be rewritten in various ways
  to give the two interpretations given in the theorem.
  Since $S_{xu|yv}(\ell)$ is equal to the length of the internal branch of the tree at $\ell$
  if it has unrooted topology $((x,u),(y,v))$ and is zero otherwise,
  \begin{align} \label{eqn:disagree}
      \mu\left(\one_{\{\ell\}} \path{x}{v} \path{y}{u} - \path{x}{u} \path_{y}{v}\right) 
      = \mu(\ell) \left( S_{xu|yv}(\ell) - S_{xv|yu}(\ell) \right),
  \end{align}
  and so expression \eqref{eqn:f4_expectation_branches} follows immediately.
  (An alternative path would be to use
  the four point condition \citep{semple2003phylogenetics},
  which implies here that
  $\path{x}{v} \path{y}{u} - \path{x}{u} \path{y}{v}
    = (\path{x}{v} + \path{y}{u} - \path{x}{u} - \path{y}{v})/2$.) 

  For the first interpretation,
  % first note that 
  % $\path{x}{v} \path{y}{u} - \path{x}{u} \path{y}{v} = 
  % \path{x}{v} + \path{y}{u} - \path{x}{u} - \path{y}{v}$.
  define $\star$ to be \emph{all} ancestral genomes alive at some fixed point in the remote past 
  (longer ago than the maximum time to most recent common ancestor across the genome),
    and note that 
    $\path{x}{u} = |\star| \path{\star}{x} + |\star| \path{\star}{u} - 2 |\star|^2 \path{\star}{x} \path{\star}{u}$. 
  (The factors of $|\star|$ are only to cancel the denominator in the definition of the path function.)
  Therefore, 
  $\path{x}{v} \path{y}{u} - \path{x}{u} \path{y}{v} 
  = |\star|^2 
  ( \path{\star}{x} - \path{\star}{y} )
  ( \path{\star}{u} - \path{\star}{v} )$.
  Since $\mu( |\star| \path{\star}{x} )(a,\ell)$ is equal to 1 if $a$ is more recent 
  than the age of the ancestors in $\star$ 
  and $x$ has inherited from $a$ at $\ell$,
  averaging over choices of $x$ gives the function $F_X$ defined in the theorem:
  $|\star| \mu( \path{\star}{X} )(a,\ell) = \mu(\ell) F_X(a,\ell)$,
  so that equation~\eqref{eqn:f4_expectation} follows.
  Note that in this interpretation, each term in the product
  is not summable on its own (e.g., $\sum_a F_X(a,\ell)=\infty$),
  but cancellation and a reasonable assumption about finiteness 
  of recent common ancestors makes differences like $F_X-F_Y$ summable.

  Finally, consider the variance.
  By Lemma \ref{thm:mean_and_cov}, 
  \begin{align}
    \var[f_4(X,Y;U,V)] 
    &= \frac{1}{G^2} \mu \left( 
    \left\{
        \path{X}{V} \path{Y}{U} - \path{X}{U} \path{Y}{V}
      \right\}^2
    \right) .
  \end{align}
  By Jensen's inequality,
  \begin{align} \label{eqn:jensen}
      \left( \path{X}{V} \path{Y}{U} - \path{X}{U} \path{Y}{V} \right)^2
      &\le
        \frac{1}{n_X n_Y n_U n_V} \sum_{x,y,u,v} \left( \path{x}{v} \path{y}{u} - \path{x}{u} \path{y}{v} \right)^2 .
  \end{align}
  Since the summand is an indicator function, 
  \begin{align}
  \mu\left(\one_{\{\ell\}} \left(\path{x}{v} \path{y}{u} - \path{x}{u} \path{y}{v} \right)^2 \right)
      &= S_{xu|yv}(\ell) + S_{xv|yu}(\ell) .
  \end{align}
  Summed over $\ell$, this gives equation \eqref{eqn:f4_variance}.

\end{proof}

Theorem \ref{thm:mean_coal_cov} describes how to estimate the covariance 
in empirical mean coalescent times between two pairs of sampled chromosomes.
Since the probability that two samples 
are heterozygous (i.e., differ) at a site is proportional to their TMRCA at that site (to first order),
it is natural to suggest the covariance of heterozygosities as an estimator
of at least a similar quantity. 
This would be arguably more natural from the population genetics point of view,
as mean heterozygosities are an easily observed statistic.
The following is therefore a useful contrast to theorem \ref{thm:mean_coal_cov}.

First, some more notation.
Let $(x_1,x_2)$ and $(y_1,y_2)$ be two pairs of sampled chromosomes,
and for each position $\ell$ in the genome,
recall that $t_{x}(\ell)$ is the total number of ancestors 
in the tree leading from $x_1$ and $x_2$ back to their common ancestor at $\ell$
(not including the common ancestor), 
and likewise for $t_y(\ell)$.
We will also need $t_{x \cap y}(\ell)$, the number of ancestors in the intersection of these two trees.
Just as $\overbar{\mu t}_x$ was defined above to be the mean time to common ancestor,
scaled by mutation rate;
so also 
$\overbar{\mu t}_{x \cap y} = (1/G) \sum_{\ell=1}^G \mu_\ell t_{x \cap y}(\ell)$.

\begin{theorem}[Empirical covariance of heterozygosities]
    \label{thm:het_covariance}
    Let $\theta_x(\ell) = 1$ if the genotypes of $x_1$ and $x_2$ differ at $\ell$,
    and $\theta_x(\ell) = 0$ otherwise, 
    and let $\bar \theta_x = (1/G) \sum_{\ell=1}^G \theta_x(\ell)$ be the mean heterozygosity of $(x_1,x_2)$;
    and likewise for $\theta_y(\ell)$ and $\bar \theta_y$.
    Then, under the infinite-sites model of mutation,
    \begin{align} \label{eqn:het_covar_mean}
        \E\left[ \left(\frac{G-1}{G^2} \sum_{\ell=1}^G \theta_x(\ell) \theta_y(\ell) \right)
        - \bar \theta_x \bar \theta_y \right]
        &= \overbar{\mu t}_{x \cap y} - \overbar{\mu t}_x \overbar{\mu t}_y 
    \end{align}
    and
    \begin{align} \label{eqn:het_covar_var}
        \var\left[ \left(\frac{G-1}{G^2} \sum_{\ell=1}^G \theta_x(\ell) \theta_y(\ell) \right)
        - \bar \theta_x \bar \theta_y \right]
        &= 
        \frac{1}{G} t_{x \cap y}(\ell) + O(1/G^3) .
    \end{align}
\end{theorem}

\begin{proof}[Proof of Theorem~\ref{thm:het_covariance}]

    % Note first that $\sum_\ell \theta_x(\ell) \theta_y(\ell)$ 
    % might be written as something like
    % \[
    %     \int_\sites 
    %     \left( \sum_{a \in \anc} \path{x_1}{x_2}(a,\ell) \muts(a,\ell) \right)
    %     \left( \sum_{b \in \anc} \path{y_1}{y_2}(b,\ell) \muts(b,\ell) \right)
    %     d\ell ,
    % \]
    % where ``$\muts(a,\ell)$'' is interpreted so that the terms in the first parenthesis 
    % give the total number of mutations differing between $x_1$ and $x_2$ at site $\ell$,
    % etcetera.
    Under the infinite-sites model, $\theta_x(\ell)$ and $\theta_y(\ell)$
    are both nonzero only if there is a mutation at $\ell$ that falls on both
    the path from $x_1$ to $x_2$ and the path from $y_1$ to $y_2$,
    and so 
    \begin{align}
        \sum_\ell \theta_x(\ell) \theta_y(\ell) = \int \path{x_1}{x_2}(a,\ell) \path{y_1}{y_2}(a,\ell) d\muts(a,\ell) .
    \end{align}
    % (Corrections could be made if the infinite-sites assumption is not a good one.)
    For convenience, let $\phi_x = \path{x_1}{x_2}$, 
    $\phi_y = \path{y_1}{y_2}$,
    and for a test function $\psi$
    write $\muts(\psi) = \int \psi dM$ and $\mu(\psi) = \int \psi d\mu$.
    In this shorthand,
    the statistic is
    \begin{align}
        \left( \frac{G-1}{G^2} \sum_{\ell=1}^G \theta_x(\ell) \theta_y(\ell) \right)
        - \bar \theta_x \bar \theta_y 
        = \frac{G-1}{G^2} \muts(\phi_x \phi_y) - \frac{1}{G^2} \muts(\phi_x) \muts(\phi_y) .
    \end{align}

    Using Lemma~\ref{thm:moments},
    $
        \E[ \muts(\phi_x \phi_y) ]
        = \mu(\phi_x \phi_y)
    $
    and
    $
        \E[ \muts(\phi_x) \muts(\phi_y) ]
        = \mu(\phi_x \phi_y) + \mu(\phi_x) \mu(\phi_y)  ,
    $
    which combine to give 
    \begin{align}
        \E\left[ \frac{G-1}{G^2} \muts(\phi_x \phi_y) - \frac{1}{G^2} \muts(\phi_x) \muts(\phi_y) \right]
        &= \frac{1}{G} \mu(\phi_x \phi_y) - \frac{1}{G^2} \mu(\phi_x) \mu(\phi_y) ,
    \end{align}
    which is equation~\eqref{eqn:het_covar_mean}.

    Again using Lemma~\ref{thm:moments},
    and the fact that since the value of $\phi_x(a,\ell)$ is everywhere either 0 or 1, 
    so $\phi_x(a,\ell)^2 = \phi_x(a,\ell)$,
    $
        \var[ \muts(\phi_x \phi_y) ]
        = \mu(\phi_x \phi_y)
    $
    and
    \begin{align}
        \cov[\muts(\phi_x \phi_y),\muts(\phi_x) \muts(\phi_y)] 
        &=
        \mu(\phi_x \phi_y) \left( 1 + \mu(\phi_x) + \mu(\phi_y) \right) 
    \end{align}
    and
    \begin{align}
        \var[\muts(\phi_x) \muts(\phi_y)]
        &=
        \mu(\phi_x \phi_y) \left( 1
            + 2 ( \mu(\phi_x) + \mu(\phi_y) ) 
            + 2 \mu( \phi_x \phi_y )
            + 4 \mu(\phi_x) \mu(\phi_y) 
        \right) 
    \end{align}
    so finally, equation~\eqref{eqn:het_covar_var} is $1/G^4$ multiplied by
    \begin{align}
        &(G-1)^2 \var[\muts(\phi_x \phi_y) ] -2 (G-1) \cov[\muts(\phi_x \phi_y),\muts(\phi_x) \muts(\phi_y)] + \var[\muts(\phi_x) \muts(\phi_y)] \\
      \begin{split}
          &\qquad \qquad = \mu(\phi_x \phi_y) \left\{
          (G-1)^2 
          - 2 (G-1) (1+\mu(\phi_x)+\mu(\phi_y)) \right. \\
          & \qquad \qquad \qquad \left. {} + 1
              + 2 ( \mu(\phi_x) + \mu(\phi_y) ) 
              + 2 \mu( \phi_x \phi_y )
              + 4 \mu(\phi_x) \mu(\phi_y) 
          \right\} .
        \end{split}
    \end{align}

\end{proof}

We now turn to the proof of Theorem \ref{thm:mean_coal_cov},
which is similar to the previous proof, but somewhat more involved.

\begin{proof}[Proof of Theorem~\ref{thm:mean_coal_cov}]

  Here, the goal is to design a statistic that estimates the covariance in mean coalescent times.
  We scale time by mutation rate because numbers of mutations only depend on the product
  of branch length and mutation rate, which therefore cannot be disentangled
  without additional assumptions (or modeling of recombination.) 
  Therefore, if we define 
  \begin{align}
    \psi_{k,x}(a,\ell) = 
    \begin{cases}
        \frac{1}{|W_k|} \path{x_1}{x_2}(a,\ell) \qquad &\text{if} \; \ell \in W_k \\
      0 \qquad &\text{otherwise}
    \end{cases}
  \end{align}
  then $\mu(\psi_{k,x})=\overbar{\mu t}_{k,x}$ gives the mean coalescent time on window $W_k$.
  We next need something whose expectation is $\mu(\psi_{k,x}) \mu(\psi_{k,y})$.
  By Lemma~\ref{thm:moments}, this is
  \begin{align}
    \E\left[M(\psi_{k,x}) M(\psi_{k,y}) - M(\psi_{k,x} \psi_{k,y})\right] = \mu(\psi_{k,x}) \mu(\psi_{k,y}) .
  \end{align}
  Similarly, if we define $\overbar{\psi}_x(a,\ell) = \frac{1}{n} \sum_{k=1}^n \psi_{k,x}(a,\ell)$,
  so that $\mu(\overbar{\psi}_x) = \overbar{\overbar{\mu t}}_x$,
  we need something whose expectation is $\mu(\overbar{\psi}_x) \mu(\overbar{\psi}_y)$,
  which by the same lemma is
  \begin{align}
    \E\left[M(\overbar{\psi}_{x}) M(\overbar{\psi}_{y}) - M(\overbar{\psi}_{x} \overbar{\psi}_{y})\right] = \mu(\overbar{\psi}_{x}) \mu(\overbar{\psi}_{y}) .
  \end{align}
  These two combine to get
  \begin{align}
    C_{x,y} &= \frac{1}{n} \sum_{k=1}^n \left\{ M(\psi_{k,x}) M(\psi_{k,y}) - M(\psi_{k,x} \psi_{k,y}) \right\} 
    - \left\{ M(\overbar{\psi}_{x}) M(\overbar{\psi}_{y}) - M(\overbar{\psi}_{x} \overbar{\psi}_{y}) \right\} .
  \end{align}
  Now note that by linearity of $M$,
  \begin{align}
    M(\overbar{\psi}_x) M(\overbar{\psi}_y) - M(\overbar{\psi}_x \overbar{\psi}_y) 
    =
    \frac{1}{n^2} \sum_{j=1}^n \sum_{k=1}^n \left( M(\psi_{j,x})M(\psi_{k,y}) - M(\psi_{j,x} \psi_{k,y}) \right) .
  \end{align}
  Defining $Z_{jk} = M(\psi_{j,x})M(\psi_{k,y}) - M(\psi_{j,x} \psi_{k,y})$,
  then
  \begin{align}
    \label{eqn:C_and_Z}
    C_{x,y} &= \frac{1}{n} \sum_{k=1}^n Z_{kk} - \frac{1}{n^2} \sum_{j=1}^n \sum_{k=1}^n Z_{jk} ,
  \end{align}
  which,
  since $M(\psi_{k,x}) = N_k(x)/|W_k|$ and $M(\psi_{k,x}\psi_{k,y}) = N_k(x,y)/|W_k|^2$,
  is equation~\eqref{eqn:coal_cov_defn}.

  The substantially lengthier calculation of the variance is postponed until the Appendix.

\end{proof}

%% Statement of results

% Application to the F4 statistic.
%   Explain f4 multilinearity used in S9 of Haak et al.

% Theorem: observability of statistics.

% Lemma: calculation of moments of such statistics.

% Theorem: consistency of estimation of moments of empirical distributions as $G \to \infty$

% Theorem: INconsistency of estimation of moments of empirical distributions as $|X| \to \infty$

%%%%% %%%%%
\subsection*{The infinite sites assumption}

Some of the simplicity above
(and more generally in population genetics)
relies on the assumption that only one mutation
can occur at each site,
which generally results in expressions that are correct up to a factor proportional to 
the fraction of sites at which more than one mutation has occurred.
To illustrate this, suppose that more than one mutation
can occur per site;
in other words, the Poisson process of mutations happens on a discrete, not continuous, set.

Assume for the moment that all mutation rates are equal: $\mu_\ell = \mu$.
Fix a pairs of sampled chromosomes $(x_1,x_2)$,
and define $\xi_n(\ell)$ is defined to be 1 if $n$ mutations have occurred 
at site $\ell$ between $x_1$ and $x_2$,
and let $\bar \xi_n = (1/G) \sum_{\ell=1}^G \xi_n(\ell)$.
The expected value of $\bar \xi_n$ under this model is
\begin{align}
  \frac{1}{G} \sum_{\ell=1}^G e^{-\mu t_x(\ell)} \frac{ (\mu t_x(\ell) )^n }{n!} ,
\end{align}
and so an estimate of $n^\mathrm{th}$ moment of the empirical distribution of $t_x(\ell)$ could be
$n! \frac{ \bar \xi_n }{ \bar \xi_0 }$.
However, since this is essentially estimating the density of sites at which there were $n$ mutations,
to have any accuracy requires a reasonable number of such sites,
and distinguishing these from sequencing error.
In practice, this is highly problematic,
and is confounded by mutation rate heterogeneity.

% \section*{Application to data}
% 
% Compute \emph{correlation} of coalescent times conditioned on sequence context
% for two Afr, two Eur, and between; compare across contexts.

%%%%% %%%%% %%%%% %%%%%
\section*{Discussion}

In this paper I have explored an empirical approach to demographic inference in population genetics, 
by treating the (empirical, realized) ancestral recombination graph
as a complex, unobserved, object we wish to learn about,
rather than as an intermediate layer that is averaged over in the course of inferring parameters of interest
in higher-level stochastic models (e.g., coalescent models).
This approach certainly does not replace coalescent theory,
but seems useful in that it can provide 
more concretely interpretable results to non-specialists
(e.g., ``numbers of common ancestors'' rather than ``coalescent rates''),
and intuition about what statistics have the best power to distinguish between alternative population models.
This way of thinking is certainly not new,
but data that give us power to infer quantities directly at this level of abstraction is.

Additionally, I have described a new formalism 
that can simplify calculations related to population genetics statistics
by writing these as integrals of functions against a Poisson random measure,
which models the locations of mutations on the genomes of all possible ancestors.

This point of view, and this formalism,
leads to two general sorts of interpretations for a given statistic:
first, in terms of the distribution of trees along which the samples are related at each locus;
and second, in terms of weighted sums over ancestral chromosomes.
This duality is easy to see considering simple examples such as pairwise divergences,
and is less obvious for more complex statistics such as $f_4$.

\paragraph{$f_4$ and family}
Theorem \ref{thm:f_4} gives an interpretation of the $f_4$ statistic
as a sum of products of differences in ancestry, over all ancestral genomes.
\citet{patterson2012ancient} interprets this and related statistics in terms of
\emph{shared drift} along branches of an \emph{admixture graph}
(in which each edge is a randomly mating population).
Since genetic drift is determined by the sharing of common ancestors,
Theorem \ref{thm:f_4} can be seen as a more precise statement of the same observation.
Since the main point of this paper is to lay out the general framework,
I have not undertaken the task of recasting the entire family of related statistics 
(e.g., ABBA-BABA),
but the general way forward can be seen by analogy to $f_4$.

\paragraph{The unknown network of ancestors}
The problem at hand, 
to infer aspects of the unknown pedigree through which genetic material has been noisily inherited,
has some parallels to the problem of \emph{active network tomography}
\citep[reviewed in][]{lawrence2006network}:
given information on losses and delays of ``probe'' packets sent
between a subset of peripheral nodes in a network,
infer the topology of the network
and certain characteristics of its internal nodes.
Results in this field on identifiability \citep[e.g.,][]{singhal2007identifiability,gopalan2012identifying}
would be very interesting in this context,
although genomic data are substantially noisier, more sparsely collected, and more numerous.
Others have already made use of the parallels to phylogenetics,
in the other direction \citep{ni2011network}.

\paragraph{Recombination}
A glaring omission in this paper is any treatment of linkage between sites:
for instance, when calculating variances, each site is treated as \emph{independent}; 
surely this cannot be right,
as the trees at nearby sites are highly correlated?
But, we begin by assuming that the mutation process at each site is independent given the trees
(which seems for the most part reasonable),
and take the entire empirical \emph{ancestral recombination graph} as fixed,
which includes the locations of ancestral recombination breakpoints.
If the aim is to estimate descriptive statistics of the ARG then this is the correct point of view;
but an intermediate position would be to treat only the \emph{population pedigree} as given,
and to additionally model randomness in recombination.
This could allow for recovery of additional information left behind by recombination,
and was the point of view taken in \citet{ralph2013geography}.
This will be the subject of future work.

\paragraph{Mutation}
Above, we have allowed for heterogeneity in mutation rate,
as this is known to be substantial \citep[e.g.,][]{misawa2009evaluation}.
When applying these methods, 
it seems necessary to choose regions of the genome with comparable large-scale mutation rates,
and then within these, to average over enough sites that small-scale heterogeneity 
will not hopelessly confound comparisons.
Methods such as described in \citet{lipson2015calibrating},
may be useful in disentangling mutation rate heterogeneity from heterogeneity in demographic histories
along the genome.
Furthermore, this paper only models diallelic markers, 
and completely ignores both sequence context
and backmutation.
This is common in population genetics methods,
since the infinite sites assumption should be a good one at the typical levels of divergence encountered within species,
especially if mutation rate heterogeneity in accounted for.
Reconciliation of more realistic models of mutation
with the Poisson process method described here
would require more work.

\paragraph{Selection}
An important assumption behind this method is that mutations are independent of the ARG.
This is not true for alleles under selection,
but neither does this approach entirely assume neutrality:
since the ARG is taken as given,
if we knew which sites were under selection and excluded these from the analyses,
then the assumptions would be satisfied.
In other words, linked, selected sites are not an issue,
as they act by distorting nearby gene trees.
On the other hand, if a large fraction of segregating polymorphisms
are actively under selection, this violates the model
(but it is unclear what to replace it with).

\subsection*{Acknowledgements}
Thanks to Paul Marjoram, Graham Coop, and Yaniv Brandvain for useful discussion,
to John Wakeley for insightful comments, 
and to an anonymous reviewer for correcting an error in the proof Theorem \ref{thm:f_4}
and otherwise greatly improving the manuscript.
This work was supported by NSF grant \#1262645 (DBI)
and startup funds from USC.

%%%%%%%%%%%%%
\bibliography{popgen}

\appendix

\section*{The variance of $C_{x,y}$}

\begin{proof}[Continuation of proof of Theorem \ref{thm:mean_coal_cov}]

  First note that \eqref{eqn:C_and_Z} implies that
  \begin{align}
    C_{x,y} &= \frac{1}{n} \sum_{j=1}^n \sum_{k=1}^n Z_{jk} \left(\delta_{jk} - \frac{1}{n}\right) ,
  \end{align}
  where $\delta_{jk}=1$ if $j=k$ and is zero otherwise.
  To find the variance of $C_{x,y}$, we need to compute covariances of the $Z_{jk}$ terms.
  To do this, it is most efficient to record the general case.
  The following lemma follows directly from Lemma~\ref{thm:moments},
  which uses upper-case roman characters for test functions 
  to make the result more visually intuitive and easier to apply.

  \begin{lemma} \label{thm:higher_cov}
    If $A$, $B$, $C$, and $D$ are test functions as in Lemma~\ref{thm:moments},
    then
  \begin{align}
    \begin{split}
      \cov[ M(A)M(B), M(C) ]
        &= 
        \mu(A B C) % (ABC)
        + \mu(A) \mu(B C) % (A)(BC)
        + \mu(B) \mu(A C) % (B)(AC)
    \end{split} 
  \end{align}
  and
  \begin{align}
    \begin{split}
      \cov[ M(A)M(B), M(C)M(D) ]
        &= 
        \mu(A B C D) % (ABCD)
        + \mu(A) \mu(B C D) % (A)(BCD)
        + \mu(B) \mu(A C D) % (B)(ACD)
        + \mu(C) \mu(A B D) % (C)(ABD)
        \\ &\qquad {}
        + \mu(D) \mu(A B C) % (D)(ABC)
        + \mu(A C) \mu(B D) % (AC)(BD)
        + \mu(A D) \mu(B C) % (AD)(BC)
        \\ &\qquad {}
        + \mu(A) \mu(C) \mu(B D) % (A)(C)(BD)
        + \mu(A) \mu(D) \mu(B C) % (A)(D)(BC)
        \\ &\qquad {}
        + \mu(A C) \mu(B) \mu(D) % (AC)(B)(D)
        + \mu(A D) \mu(B) \mu(C) % (AD)(B)(C)
    \end{split} .
  \end{align}
\end{lemma}

\details{
\begin{proof}[of Lemma~\ref{thm:higher_cov}]
  \begin{align}
    \begin{split}
      \cov[ M(A)M(B), M(C)M(D) ]
        &= 
        \mu(A B C D ) % (ABCD)
        + \mu(A) \mu(B C D) % (A)(BCD)
        + \mu(B) \mu(A C D) % (B)(ACD)
        \\ &\qquad {}
        + \mu(C) \mu(A B D) % (C)(ABD)
        + \mu(D) \mu(A B C) % (D)(ABC)
        + \mu(A B) \mu(C D) % (AB)(CD)
        \\ &\qquad {}
        + \mu(A C) \mu(B D) % (AC)(BD)
        + \mu(A D) \mu(B C) % (AD)(BC)
        + \mu(A) \mu(B) \mu(C D) % (A)(B)(CD)
        \\ &\qquad {}
        + \mu(A) \mu(C) \mu(B D) % (A)(C)(BD)
        + \mu(A) \mu(D) \mu(B C) % (A)(D)(BC)
        + \mu(A B) \mu(C) \mu(D) % (AB)(C)(D)
        \\ &\qquad {}
        + \mu(A C) \mu(B) \mu(D) % (AC)(B)(D)
        + \mu(A D) \mu(B) \mu(C) % (AD)(B)(C)
        + \mu(A) \mu(B) \mu(C) \mu(D) % (A)(B)(C)(D)
        \\ &\qquad {}
        - \left\{ 
        \mu(A B) + \mu(A) \mu(B)
        \right\} % (AB)+(A)(B)
        \left\{ 
        \mu(C D) + \mu(C) \mu(D)
        \right\} % (CD)+(C)(D)
    \end{split} 
  \end{align}
  and
  \begin{align}
    \begin{split}
      \cov[ M(A)M(B), M(C) ]
        &= 
        \mu(A B C ) % (ABC)
        + \mu(A) \mu(B C) % (A)(BC)
        + \mu(B) \mu(A C) % (B)(AC)
        \\ &\qquad {}
        + \mu(C) \mu(A B) % (C)(AB)
        + \mu(A) \mu(B) \mu(C) % (A)(B)(C)
        \\ &\qquad {}
        - \left\{ 
        \mu(A B) + \mu(A) \mu(B)
        \right\} % (AB)+(A)(B)
        \mu(C)   % (C)
    \end{split} 
  \end{align}
\end{proof}
}

If $j \neq k$,
then $\psi_{j,x}$ and $\psi_{k,x}$ are supported on disjoint parts of the ancestral genomes,
and so $\psi_{j,x} \psi_{k,y} = 0$.
For the same reason, 
the covariance of $Z_{jk}$ and $Z_{\ell m}$ is zero unless at least one of $j$ and $k$ match at least one of $\ell$ and $m$.
We are further helped by the fact that $\psi_{j,x}^n = \psi_{j,x} / W_j^{n-1}$.

First, we apply Lemma~\ref{thm:higher_cov} with $A=C=\psi_{j,x}$ and $B=D=\psi_{k,y}$
and $j \neq k$
(so, note that $AB=AD=CB=CD=0$).
\begin{align}
  % \var[M(\psi_{j,x})M(\psi_{k,y})] 
  \var[Z_{jk}]
  &=
  \frac{ \mu(\psi_{j,x}) \mu(\psi_{k,y}) }{ |W_j| |W_k| } 
    \left\{ 1  + |W_j| \mu(\psi_{j,x}) + |W_k| \mu(\psi_{k,y}) \right\}  % CHECKED
\end{align}
Similarly,
if $A=\psi_{j,x}$, $B=\psi_{k,y}$, $C=\psi_{k,x}$, and $D=\psi_{j,y}$, 
and $j \neq k$
(so, note that $AB=AC=BD=CD=0$).
\begin{align}
    \begin{split}
      % \var[M(\psi_{j,x})M(\psi_{k,y})] 
      \cov[Z_{jk}, Z_{kj}]
      &= 
        \mu(\psi_{j,x} \psi_{j,y}) \mu(\psi_{k,y} \psi_{k,x}) 
        + \mu(\psi_{j,x}) \mu(\psi_{j,y}) \mu(\psi_{k,y} \psi_{k,x})
        + \mu(\psi_{j,x} \psi_{j,y}) \mu(\psi_{k,y}) \mu(\psi_{k,x})   % CHECKED
    \end{split}
\end{align}

Now, with $A=C=\psi_{j,x}$, $B=\psi_{k,y}$, and $C=\psi_{l,y}$,
and $j$, $k$ and $l$ all distinct,
\begin{align}
  \cov[Z_{jk},Z_{jl}]
  &=
  \frac{1}{W_j} \mu(\psi_{j,x}) \mu(\psi_{k,y}) \mu(\psi_{l,y})   % CHECKED
\end{align}
and similarly,
\begin{align}
  \cov[Z_{kj},Z_{jl}]
  &=
  \mu(\psi_{j,x} \psi_{j,y}) \mu(\psi_{k,x}) \mu(\psi_{l,y}) .
\end{align}
  % \details{ % BEGIN DETAILS
    Now, with $A=C=\psi_{j,x}$, $B=\psi_{j,y}$, and $D=\psi_{k,y}$,
    \begin{align}
      \cov[ M(\psi_{j,x})M(\psi_{j,y}), M(\psi_{j,x})M(\psi_{k,y}) ]
      &=
      \mu( \psi_{k,y} ) \left\{
        \mu( \psi_{j,x} \psi_{j,y} ) 
        \left( 
          \frac{1}{W_j} 
          + \mu( \psi_{j,x} )
        \right)
        + \frac{1}{W_j} \mu( \psi_{j,x} ) \mu( \psi_{j,y} )
      \right\}   % CHECKED
    \end{align}
    and with $A=\psi_{j,x}$, $B=\psi_{k,y}$, and $C=\psi_{j,x}\psi_{j,y}$,
    \begin{align}
     \cov[M(\psi_{j,x}) M(\psi_{k,y}), M(\psi_{j,x}\psi_{j,y}) ]
     &=
     \frac{1}{W_j} \mu(\psi_{k,y})\mu(\psi_{j,x}\psi_{j,y}) .  % CHECKED
    \end{align}
  % } % END DETAILS
These combine to give
\begin{align}
  \cov[Z_{jj},Z_{jk}]
  \details{
  &=
    \cov[ M(\psi_{j,x})M(\psi_{j,y}), M(\psi_{j,x})M(\psi_{k,y}) ] 
    - \cov[M(\psi_{j,x}) M(\psi_{k,y}), M(\psi_{j,x}\psi_{j,y}) ] \\
  }
  &= 
      \mu( \psi_{k,y} ) \left\{
        \mu( \psi_{j,x} \psi_{j,y} ) 
          \mu( \psi_{j,x} )
        + \frac{1}{W_j} \mu( \psi_{j,x} ) \mu( \psi_{j,y} )
      \right\}   % CHECKED
\end{align}

Finally, taking $A=C=\psi_{j,x}$ and $B=D=\psi_{j,y}$, we get
\begin{align}
  \begin{split}
    \var[ M(\psi_{j,x}) M(\psi_{j,y}) ]  % A = C = psi_x; B = D = psi_y
        &= 
        \frac{1}{W_j^2} \mu(\psi_{j,x} \psi_{j,y}) 
        + \frac{2}{W_j} \mu(\psi_{j,x} \psi_{j,y})
        \left( \mu(\psi_{j,x}) + \mu(\psi_{j,y}) \right)
        \\ &\qquad {}
        + \frac{1}{W_j^2} \mu(\psi_{j,x}) \mu(\psi_{j,y})
        + \mu(\psi_{j,x} \psi_{j,y})^2
        + \frac{1}{W_j} \mu(\psi_{j,x})^2 \mu(\psi_{j,y})
        \\ &\qquad {}
        + \frac{1}{W_j} \mu(\psi_{j,x}) \mu(\psi_{j,y})^2
        + 2 \mu(\psi_{j,x}) \mu(\psi_{j,y}) \mu(\psi_{j,y} \psi_{j,x})
    \end{split}   % CHECKED
\end{align}
and
\begin{align}
  \begin{split}
    \cov[ M(\psi_{j,x}) M(\psi_{j,y}), M(\psi_{j,x} \psi_{j,y}) ] 
    &= % A = \psi_{j,x} \psi_{j,y}, B = \psi_{j,x}, C = \psi_{j,y}
        \frac{1}{W_j^2} \mu(\psi_{j,x}\psi_{j,y}) 
        + \frac{1}{W_j} \mu(\psi_{j,x} \psi_{j,y}) 
        \left( \mu(\psi_{j,x}) + \mu(\psi_{j,y} \right)
  \end{split}  % CHECKED
\end{align}
and
\begin{align}
  \var[ M(\psi_{j,x} \psi_{j,y}) ]
  &=
  \frac{1}{W_j^2} \mu( \psi_{j,x} \psi_{j,y} )  % CHECKED
\end{align}
so
\begin{align}
  \var[Z_{jj}] 
  &= 
  \var[ M(\psi_{j,x}) M(\psi_{j,y}) ] 
  - 2 \cov[ M(\psi_{j,x}) M(\psi_{j,y}), M(\psi_{j,x} \psi_{j,y}) ] 
  + \var[ M(\psi_{j,x} \psi_{j,y}) ] \\  % CHECKED
  \begin{split}
        &= 
        \mu(\psi_{j,x}) \mu(\psi_{j,y})
        \left(
          \frac{1}{W_j^2} + \frac{1}{W_j} (\mu(\psi_{j,x}) + \mu(\psi_{j,y}))
          + 2 \mu(\psi_{j,y} \psi_{j,x})
        \right)
        + \mu(\psi_{j,y} \psi_{j,x})^2
    \end{split}   % CHECKED
\end{align}

We can put these together to obtain that
  \begin{align}
    \label{eqn:full_C_variance}
    \var[C_{x,y}] 
    &= 
    \frac{1}{n^2} \sum_{1 \le j,k,\ell,m \le n} \cov[Z_{jk},Z_{\ell m}] \left(\delta_{jk} - \frac{1}{n}\right)\left(\delta_{\ell m} - \frac{1}{n}\right) \\
    \begin{split}
      &= 
      \frac{(n-1)^2}{n^4} \sum_j \var[Z_{jj}]
      - \frac{2(n-1)}{n^4} \sum_{j \neq k} \cov[Z_{jj}, Z_{jk} + Z_{kj}]
      + \frac{1}{n^4} \sum_{j \neq k} \var[Z_{jk}]
      \\&\qquad {}
      + \frac{1}{n^4} \sum_{j \neq k} \cov[Z_{jk},Z_{kj}]
      + \frac{1}{n^4} \sum_{j \neq k \neq \ell \neq j} \cov[Z_{jk},Z_{j\ell}] 
      \\&\qquad {}
      + \frac{1}{n^4} \sum_{j \neq k \neq \ell \neq j} \cov[Z_{kj},Z_{\ell j}] 
      + \frac{2}{n^4} \sum_{j \neq k \neq \ell \neq j} \cov[Z_{kj},Z_{j\ell}] 
    \end{split} \\  % CHECKED
    \begin{split}
      &= 
      \frac{(n-1)^2}{n^4} \sum_j % \var[Z_{jj}]
        \left(
          \mu(\psi_{j,x}) \mu(\psi_{j,y})
          \left(
            \frac{1}{W_j^2} + \frac{1}{W_j} (\mu(\psi_{j,x}) + \mu(\psi_{j,y}))
            + 2 \mu(\psi_{j,y} \psi_{j,x})
          \right)
          + \mu(\psi_{j,y} \psi_{j,x})^2
        \right)
      \\&\qquad {}
      - \frac{2(n-1)}{n^4} \sum_{j \neq k} % \cov[Z_{jj}, Z_{jk} + Z_{kj}]
        \left(
          \mu( \psi_{k,y} ) \left\{
            \mu( \psi_{j,x} \psi_{j,y} ) 
              \mu( \psi_{j,x} )
            + \frac{1}{W_j} \mu( \psi_{j,x} ) \mu( \psi_{j,y} )
          \right\} 
        \right.  \\&\qquad \qquad \left. {}
          +
          \mu( \psi_{k,x} ) \left\{
            \mu( \psi_{j,y} \psi_{j,x} ) 
              \mu( \psi_{j,y} )
            + \frac{1}{W_j} \mu( \psi_{j,y} ) \mu( \psi_{j,x} )
          \right\} 
        \right)
      \\&\qquad {}
      + \frac{2}{n^4} \sum_{j \neq k} % \var[Z_{jk}]
        \left(
          \frac{1}{W_j} \mu(\psi_{j,x}) \mu(\psi_{k,y}) \left\{ \frac{1}{W_j}  + \mu(\psi_{j,x}) + \mu(\psi_{k,y}) \right\}
        \right)
      \\&\qquad {}
      + \frac{1}{n^4} \sum_{j \neq k} % \cov[Z_{jk},Z_{kj}]
        \left(
          \mu(\psi_{j,x} \psi_{j,y}) \mu(\psi_{k,y} \psi_{k,x}) 
          + \mu(\psi_{j,x}) \mu(\psi_{j,y}) \mu(\psi_{k,y} \psi_{k,x})
          + \mu(\psi_{j,x} \psi_{j,y}) \mu(\psi_{k,y}) \mu(\psi_{k,x}) 
        \right)
      \\&\qquad {}
      + \frac{1}{n^4} \sum_{j \neq k \neq \ell \neq j} % \cov[Z_{jk},Z_{j\ell}] 
          \frac{1}{W_j} \mu(\psi_{j,x}) \mu(\psi_{k,y}) \mu(\psi_{l,y}) 
      \\&\qquad {}
      + \frac{1}{n^4} \sum_{j \neq k \neq \ell \neq j} % \cov[Z_{jk},Z_{j\ell}] 
          \frac{1}{W_j} \mu(\psi_{j,y}) \mu(\psi_{k,x}) \mu(\psi_{l,x}) 
      \\&\qquad {}
      + \frac{2}{n^4} \sum_{j \neq k \neq \ell \neq j} % \cov[Z_{kj},Z_{j\ell}] ,
          \mu(\psi_{j,x} \psi_{j,y}) \mu(\psi_{k,x}) \mu(\psi_{l,y}) .
    \end{split} 
  \end{align}

  The above is dominated by terms of the form $(\overbar{\mu t})^3 / n |W|$,
  where $\overbar{\mu t}$ is an average of $\mu(\psi_{\cdot})$ terms.
  To put a crude bound on this,
  assume that $\mu(\psi_{j,x})$ and $\mu(\psi_{j,y})$ are bounded by $\epsilon$ for every $j$,
  resulting in $\var[C_{x,y}] \le \frac{\epsilon^3}{n|W|}$,
  as in the theorem.

\end{proof}

\end{document}